\documentclass[twoside,leqno,twocolumn]{article}

\usepackage[letterpaper]{geometry}
\usepackage{ltexpprt}

\usepackage{mathtools}

\usepackage{amsmath,amssymb}
\usepackage{graphicx}
\usepackage{hyperref}

\DeclareMathOperator{\vol}{Vol}

\title{Exact computation of a manifold metric, via Lipschitz Embeddings and Shortest Paths on a Graph}
\author{
  Timothy Chu$^*$ \\
  CMU\\
  \texttt{timothyzchu@gmail.com}
  \and
  Gary L.\ Miller\thanks{Partially supported by NSF CCF-1637523}\\
  CMU\\
  \texttt{glmiller@cs.cmu.edu} \\
  \and
    \newline
  Donald R. Sheehy\thanks{Partially supported by NSF CCF-1652218}\\
  North Carolina State University \\
  \texttt{don.r.sheehy@gmail.com}
}
\date{}

\newcommand{\eps}{\varepsilon}

\newcommand{\B}{\mathbb{B}}

\newcommand{\len}{\ell}
\newcommand{\R}{\mathbb{R}}
\newcommand{\ourpath}{\mathrm{path}}
\newcommand{\dist}{\mathbf{d}}
\newcommand{\distto}{\mathbf{r}}
\renewcommand{\because}[1]{&\left[\text{\small{#1}}\right]}

\newcommand\prob[2]{\operatorname*{\mathbb{P}}_{#1}\left[ #2 \right]}
\newcommand{\M}{M}

\newtheorem{prop}{Proposition}[section]
\newtheorem{definition}{Definition}[section]

\usepackage{color}

\begin{document}

  \maketitle

  \fancyfoot[R]{\scriptsize{Copyright \textcopyright\ 2020 by SIAM\\
  Unauthorized reproduction of this article is prohibited}}

  \begin{abstract}

  % Basic Object and main result
  Data-sensitive metrics adapt distances locally based the density of data points with the goal of aligning distances and some notion of similarity.
  In this paper, we give the first exact algorithm for computing a data-sensitive metric called the nearest neighbor metric.
  In fact, we prove the surprising result that a previously published $3$-approximation is an exact algorithm.

  % Surprising and Hard
  The nearest neighbor metric can be viewed as a special case of a density-based distance used in machine learning, or it can be seen as an example of a manifold metric.
  Previous computational research on such metrics despaired of computing exact distances on account of the apparent difficulty of minimizing over all continuous paths between a pair of points.

  % ancillary results
  We leverage the exact computation of the nearest neighbor metric to compute sparse spanners and persistent homology.
  We also explore the behavior of the metric built from point sets drawn
from an underlying distribution and consider the more general case of
inputs that are finite collections of path-connected compact sets.

  % Interesting Connections
  The main results connect several classical theories such as the conformal change of Riemannian metrics, the theory of positive definite functions of Schoenberg, and screw function theory of Schoenberg and Von Neumann.
  We also develop some novel proof techniques based on the combination of screw functions and Lipschitz extensions that may be of independent interest.

\end{abstract}

  %\clearpage
  %  \input{old_introduction}
  \section{Introduction}

The profound success of nonlinear methods in machine learning such as kernels methods, density-based distances, and neural nets reveals that although data are often represented as points in $\R^n$, the shortest path between two points is \emph{not} a straight line.
It is widely believed that a more useful metric on the data points would
have the property that two points in a dense cluster will be close in some
underlying metric, even if the Euclidean distance is
far~\cite{alamgir12shortest,
cohen15approximating, vincent03, bijral11semiSupLearningDBD}.
That is, distances are scaled inversely according to the density of the data along a path between points.
We call such a metric \textbf{data-sensitive}.

Data-sensitive metrics arise naturally in machine learning, and are
implicitly central in celebrated methods such as $k$-NN graph methods,
manifold learning, level-set methods, single-linkage clustering, and
Euclidean MST-based clustering (see Section~\ref{sec:edge-power} and
Appendix~\ref{ap:clustering-link} for details).
%(See Appendix~\ref{} for details)
The construction of appropriate data-sensitive metrics is an active area of research.
We consider a simple data-sensitive metric with an underlying manifold structure called the \textbf{nearest neighbor metric}.
This metric was first introduced in~\cite{cohen15approximating}. It and its close variants have been studied in the past by
multiple
researchers~\cite{hwang2016, cohen15approximating, sajama05estimatingDBDM,
bijral11semiSupLearningDBD,vincent03}.
In this paper, we show how to compute the nearest neighbor metric exactly
for any dimension, which solves one of the most important and challenging
problems for any manifold-based metric.

The starting point will be the nearest neighbor function $\distto_P$ for
the data set $P$:
\[
\distto_P(z) = 4\min_{x\in p} \|x-z\|,
\]
where the factor of $4$ normalizes and simplifies expressions later.
This function is also known as the distance function to the set $P$ and is the basic object of study in the critical point theory of distance functions, a generalization of Morse Theory~\cite{grove93critical}.
This theory has found many recent uses in computational geometry~\cite{chazal08smooth,chazal09sampling} as it is a natural way to infer underlying structure from a sample of points.
We have a similar goal of inferring underlying structure when we use $\distto_P$ as a cost function for a density-based distance defined as follows (see also Section~\ref{sec:persistence} for explicit inference results).

\begin{definition}
Given a continuous cost function $c:\R^k \rightarrow \R$, we define the density-based
cost of a path $\gamma$ relative to $c$ as:
\[ \len_c(\gamma) = \int_{0}^1 c(\gamma(t)) \| \gamma'(t) \|dt. \]
Here, the path $\gamma$ is defined as a continuous map $\gamma:[0,1]
\to \R^k$.
Let $\ourpath(a,b)$ denote the set of piecewise-$C_1$ paths from $a$ to $b$.
We then define the \textbf{density-based distance} between two points $a, b \in
\R^k$ as
\[ d_c(a,b) \inf_{\gamma\in\ourpath(a,b)} \len_c(\gamma)\]
\end{definition}

This is a slight simplification of the density-based distances from~\cite{sajama05estimatingDBDM} which included other requirements to facilitate approximation.
Conceptually, the density-based cost of a path is the weighted path length,
where each infinitesimal path piece is weighted according to $c$.  The cost
$c$ is usually some function of an underlying density $f$ (the natural
choice would be $c(x) = f(x)^{-\frac{1}{k}}$).  Density-based distances
have been notable in the machine learning setting for over a
decade~\cite{sajama05estimatingDBDM,bijral11semiSupLearningDBD}.  To build
a data-sensitive metric from density-based distances, we would like a
cost function $c$ that is small when close to the data set, and large when
far away.  The nearest neighbor function $\distto_P$ is the most natural
candidate, and has been traditionally used as a proximity measure between
points and a data set in both the geometry and machine learning
settings~\cite{bijral11semiSupLearningDBD}. It has been used as such in
nearest neighbor (and $k$-NN) classification, $k$-means/medians/center
clustering, finite element methods, and any of the numerous methods that
use Voronoi diagrams or Delaunay triangulation as intermediate data
structures.

\begin{definition} Given any finite set $P\subset \R^k$, the \textbf{nearest neighbor cost} function is $\len_N := \len_{\distto_P}$ and the \textbf{nearest neighbor metric} is $\dist_N := \dist_{\distto_P}$.  That is, it's the density-based distance with cost function $\distto_P$.
\end{definition}
\begin{figure}[htbp]
  \centering
    \includegraphics[width=0.3\textwidth]{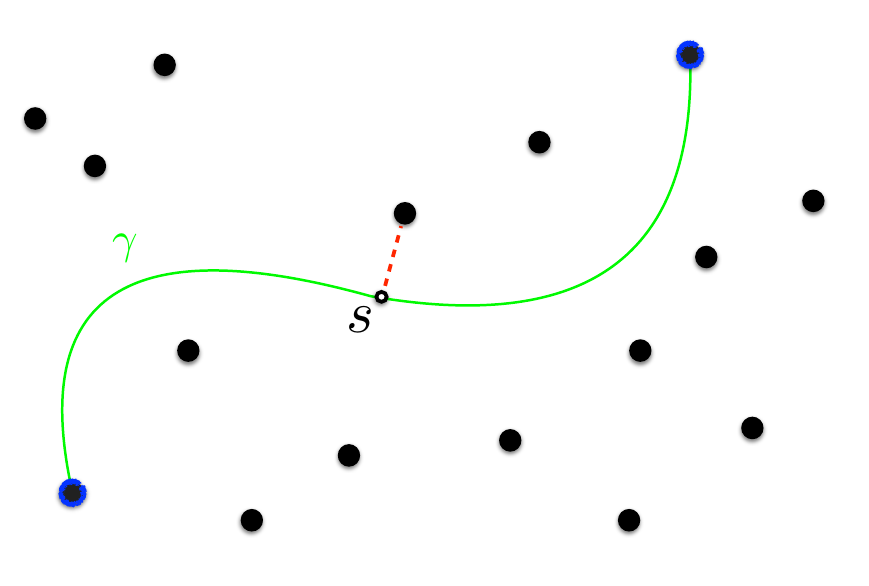}
    \caption{In this figure we have a collection of points.
      The length or cost of the green curve between the two blue points
      is the integral along the curve scaled by the distance to the nearest
point.}
  \label{fig:example}
\end{figure}

The nearest neighbor metric, and density-based distances in general, are
examples of manifold geodesics~\cite{sajama05estimatingDBDM,
tenenbaum00global}.  Manifold geodesics of data sets are
defined by embedding points into a manifold and computing the infimum
length path in the manifold.  Within computer science, dozens of
foundational papers in machine learning and surface reconstruction rely on
manifold-based metrics to perform clustering, classification, regression,
surface reconstruction, persistent homology, and
more~\cite{tenenbaum00global, cohen15approximating, vincent03,
bijral11semiSupLearningDBD,  sajama05estimatingDBDM,
edelsbrunner02topological, alamgir12shortest, vL09}.  Manifold geodesics
predate computer science, and are the cornerstone of many fields of physics
and mathematics.  Exactly computing geodesics is fundamental to countless
areas of physics including: the brachistochrone and minimal-drag-bullet
problem of Bernoulli and Newton~\cite{bernoulli}, exactly determining a particle's
trajectory in classical physics (Hamilton's Principle of Least
Action)~\cite{Courant53}, computing the path of light through a
non-homogeneous medium (Snell's law), finding the evolution of wave
functions in quantum mechanics over time (Feynman path integrals~\cite{Feynman48}), and
determining the path of light in the presence of gravitational fields
(General Relativity, Schwarzschild metric)~\cite{Schwarzschild, Sussmann97}. In
mathematics, manifold geodesics appear in many branches of higher
mathematics including differential equations, differential geometry, Lie
theory, calculus of variations, algebraic geometry, and topology.

One of the most significant problems on any manifold geodesic is how to
compute its length.  Exact computation of manifold metrics is considered a
fundamental problem in mathematics and physics, dating back for four
centuries: entire fields of mathematics, including the celebrated calculus
of variations, have arisen to tackle this~\cite{Courant53}. Historically,
mathematicians placed strong emphasis on exact computation as opposed to
constant factor approximations~\cite{Courant53}. An algorithmic problem on manifold
geodesics, with modern origins, is to $(1+\eps)$ approximate these metrics
efficiently on a computer.  The core difficulty in the first problem is
that geodesics are the minimum cost path out of an uncountable number of
paths that can travel 'anywhere' on the manifold structure.  This makes
exactly computing these metrics challenging, even in the case of the
nearest neighbor metric for just four fixed points in two dimensions (the
authors are unaware of any easy method for this simplified task).
% The core tool for exactly computing manifold metrics, calculus of
% variations, is intractable on the nearest neighbor metric due to the
% metric's heavy dependence on the Voronoi diagram of the point set, which
% can be quite complicated for even five points in two dimensions (for more
% on this approach and its limitations, see []).
Calculus of variations can show that the optimal nearest neighbor path is
piecewise hyperbolic, but this is generally insufficient to exactly compute
the nearest neighbor metric---there are point sets where there are
many differentiable, piecewise hyperbolic paths between two data points with
different costs.

In this paper, we solve both problems: we exactly compute the Nearest
Neighbor metric in all cases, and we $(1+\eps)$ approximate it quickly.
Our approach is based on a novel embedding of the data into high dimensions where the geodesics are straight lines.
Then we use a Lipschitz extension theorem to relate the lengths of the shortest paths in the original space and the embedding.
We combine these tools to prove that the nearest neighbor metric is exactly equal to a shortest path distance on a geometric graph, the so-called edge-squared metric, in all cases.
This allows us to compute the nearest-neighbor metric exactly for any given point set in polynomial time, and it is the only known (non-trivial) density-based distance that can be computed by a discrete algorithm.

  \begin{definition}
For $x\in \R^d$, let $\|x\|$ denote the Euclidean norm.
  For a set of points $P\subset \R^d$:
  the \textbf{edge-squared metric} for $a,b\in P$ is
  \[
    \dist_2(a,b) = \inf_{(p_0,\ldots, p_k)}\sum_{i=1}^k \|p_i - p_{i-1}\|^2,
  \]
  where the infimum is over sequences of points $p_0,\ldots, p_k\in P$ with $p_0 = a$ and $p_k = b$.
  \end{definition}

\begin{theorem}\label{thm:NN} The nearest neighbor metric and edge squared
metric are equivalent for any set $P$ in arbitrary dimension that is the
finite collection of compact path-connected sets.
\end{theorem}

This in particular covers the case of $n$ points in $n-1$ dimension. The exact equality is realized when the nearest neighbor path is piecewise
linear, traveling straight from data point to data point.
The edge squared metric has been previously studied by multiple researchers
in machine learning and power-efficient wireless networks, but previously
has only been linked to the nearest neighbor metric by a fairly weak
3-approximation~\cite{cohen15approximating}.
There are several reasons why it is surprising that these metrics are equal:

\begin{enumerate}

\item The optimal nearest neighbor path for two points not in the dataset
is generally composed of hyperbolic arcs.
This holds true even when the dataset is a single point, and was established by~\cite{cohen15approximating} using tools in Riemannian surfaces and the complex plane.
Meanwhile, our Theorem implies an optimal nearest neighbor path for two
data points (in a dataset of any size) is piecewise linear!

\item There are simple and natural variants of the nearest neighbor metric, for which no analog of Theorem~\ref{thm:NN} is known nor suspected.
For example, if one considers powers (other than one) of the distance function as a cost, a corresponding graph-based metric is known to exist only for sets of size at most two.
% These variants are known as the $q$-nearest neighbor metric, for $1 < q < 2$, and we will formally define these
% metrics later in the introduction. When $q=2$, these
% metrics coincide with the nearest neighbor metric.
% This
% gives us a natural suite of metrics that smoothly converge
% to the nearest neighbor metric, for which no theorem like
% Theorem~\ref{thm:NN} is known.

\item For just three points in a right triangle configuration, there exist an uncountable suite of optimal-cost paths between the two endpoints of the hypotenuse.
Each path in this uncountable suite is piecewise hyperbolic, but, surprisingly, they all have the exact same cost as the edge-squared distance.
Thus, there shortest paths may not even be unique.
% Thus, lowering the nearest neighbor function anywhere inside the triangle and using this to build a new density-based distance will break Theorem~\ref{thm:NN}.
% This establishes that the equality in Theorem~\ref{thm:NN} is fairly tight.

 \item The finite union of compact path-connected geometric bodies in arbitrary
 dimension can have extremely complicated geometry, and the Voronoi diagram
on which the nearest neighbor metric depends is poorly understood for even
three of these bodies in two dimension.  There is no other restriction on
the compact geometric objects, and they need not be convex or even simply
connected, see figure~\ref{fig:example1}.
\begin{figure}[htbp]
  \centering
    \includegraphics[width=0.3\textwidth]{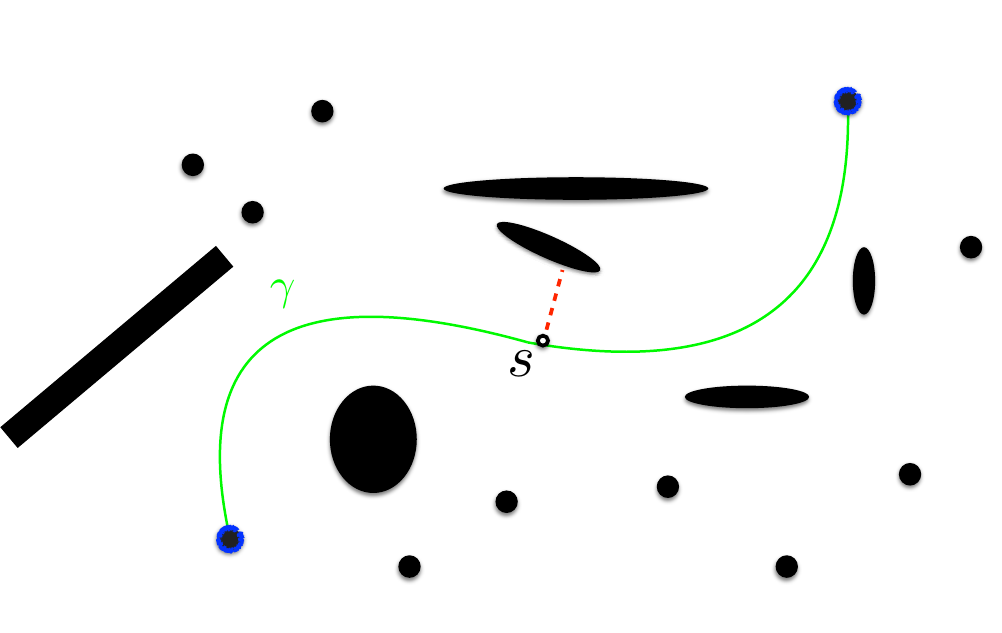}
    \caption{In this figure we have a collection of compact bodies in black.
      The length or cost of the green curve between the two blue points
      is the integral along the curve scaled by the distance to the nearest body.
    A curve may traverse a body at  no cost. Theorem~\ref{thm:NN}
establishes that the shortest path curve between two points goes straight
from compact body to compact body.}
  \label{fig:example1}
\end{figure}
\end{enumerate}

We can now tackle a second problem of interest for manifold geodesics,
which is efficiently $(1+\eps)$ approximating them. In this paper, we show
that the nearest neighbor metric admits $(1+\epsilon)$ spanners computable
in nearly-linear time, with linear size, for any point set in constant
dimension. Remarkably, these spanners are significantly sparser and faster
to compute than the theoretically optimal Euclidean spanners with the same
approximation constant, and nearly match the sparsity of the best known
Euclidean Steiner spanners. Moreover, if the point set comes from a
well-behaved probability distribution in constant dimension (a foundational
assumption in machine learning~\cite{hwang2016}), we show that the nearest neighbor
metric has perfect $1$-spanners of nearly linear size. The latter result is
impossible for many non-density sensitive metrics, such as the Euclidean
metric. Both results rely on Theorem~\ref{thm:NN}, and significantly
improve the nearest neighbor spanners of Cohen et al in~\cite{cohen15approximating}.

Theorem~\ref{thm:NN} and our spanner theorems solve two core problems of
interest for the nearest neighbor metric: exactly computing it for any
dimension, and approximating it quickly for both general point sets and
point sets arising from a well-behaved probability distribution in constant
dimension. This is the first work we know of that computes a manifold
metric exactly without calculus of variations, and we hope that our tools
can be useful for other metric computations and approximations.

% Besides for this contribution, we also generalize the nearest neighbor
% Metric to the $q$-nearest neighbor metric (abbreviated $q$-NN for short),
% and exactly compute this metric for all point sets with $\leq 4$ points for all $q>2$. We
% do this by equating it to the $q$-edge power metrics. Both the $q$-NN and
% $q$-edge power metrics will be defined later.
% \begin{theorem} \label{thm:qNN}
% For point sets that are the union of up to $4$ connected compact sets, the $q$-NN metric is exactly equal to
% the $q$-edge power metric when $q>2$. This equality is false for all $q <
% 2, q\not=1$.
% \end{theorem}
% Our equality is robust enough to handle the union of $4$ compact sets in
% any dimension. These unions
% ch can have very complicated geometry, and their Voronoi diagrams are in
% general difficult to understand. This is what makes theorems
% Theorem~\ref{thm:qNN} surprising. We further conjecture:
% \begin{conjecture}\label{conj:qNN}
% For any compact set, the $q$-NN metric is exactly equal to the $q$-edge
% poewr metric when $q>2$.
% \end{conjecture}
% If true, this would give us a quadratic algorithm to compute the $q$-NN
% metric for any $n$ point set.

% We further show that $q$-edge power metrics (and thus, it is hoped, the $q$-nearest neighbor metrics) are natural generalizations of maximum-edge-length distances on Euclidean MSTs, which in turn are fundamental for celebrated clustering methods like single-linkage clustering~\cite{}.
% This implies that the $q$-edge power metric, and the nearest neighbor metric, can be used to generalize popular methods in clustering.

  \subsection{Contributions and Past Work}
Our primary contribution is Theorem~\ref{thm:NN}, which lets us exactly
compute the nearest neighbor metric. This significantly strengthens a core
result of Cohen et al~\cite{cohen15approximating}. This theorem should be considered
quite surprising: it equates the nearest neighbor metric with the
edge-squared metric, even when the point set is a collections of compact,
path-connected objects in arbitrarily large dimension. There are no
restrictions on the convexity or simple-connectedness of such objects, so
in general the Voronoi diagram of these objects (on which the nearest
neighbor metric critically depends) can be extremely complicated.

 Besides for exactly computing the nearest neighbor metric, we present the
following theorems on approximate computation:

 %% \begin{theorem} \label{thm:NN} Given a point set $P \in \mathbb{R}^d$, the edge-squared metric on $P$
%%   and the nearest neighbor metric on $P$ are always equivalent.
%% \end{theorem}

\begin{theorem} \label{thm:general-spanner}
  For any set of points in $\mathbb{R}^d$ for constant $d$, there exists a $(1+\eps)$
  spanner of the nearest neighbor metric
  with size $O\left(n\eps^{-d/2} \right)$ computable in time
  $O\left(n \log n + n\eps^{-d/2}\log{\frac{1}{\eps}}\right)$. The
  $\log{\frac{1}{\eps}}$ term goes away given access to an algorithm
computing floor function
in $O(1)$ time.
\end{theorem}

\begin{theorem} \label{thm:distribution-spanner}
Suppose points $P$ in Euclidean space are drawn i.i.d from a Lipschitz probability density bounded
above and below by a constant, with support on a
smooth, connected, compact manifold with intrinsic dimension $d$ with boundary of bounded curvature. Then w.h.p. the $k$-NN graph of
  $P$ for $k = O(2^d \ln n)$ and edges weighted with Euclidean
  distance squared, is a $1$-spanner of the nearest neighbor
  metric on $P$.
\end{theorem}

These theorems rely on Theorem~\ref{thm:NN} and considerably strengthen the
spanner results on the nearest neighbor metric
from~\cite{cohen15approximating}. They critically rely on
Theorem~\ref{thm:NN}, which show it suffices to compute spanners of the
edge-squared metric.
Previously, sparse spanners of the edge-squared metric were shown to exist in two
dimensions via Yao graphs and Gabriel graphs~\cite{LiWan2001}, but these
did not generalize well to constant dimension: Yao
graphs are not very efficient to compute, and Gabriel graphs can have
quadratically many edges even in $3$ dimensions~\cite{chazelle94selecting}.
The spanners we produce are sparser than the
theoretical optimal for Euclidean spanners~\cite{Le19}.

Theorem~\ref{thm:distribution-spanner} proves that a $1$-spanner of
the nearest neighbor metric can be found assuming points are samples from a
probability density, by using a $k$-$NN$ graph for
appropriate $k$. Our result is tight when $d$ is constant. This
is not possible for Euclidean distance, as a $1$-spanner is almost
surely the complete graph. Although the restrictions on the probability
density may seem limiting,
they are in fact quite flexible and standard in
machine learning theory and practice~\cite{hwang2016, alamgir12shortest}. For example, although they do not cover the case
of a Gaussian (unbounded support), they do cover the case of a Gaussian
where the very thin tail is cut off, and this recovers most of the relevant
data in a Gaussian distribution. Past work on similar results
include~\cite{Balister05, Gonzales2003}.

Theorem~\ref{thm:NN} will additionally allow us
to compute the persistent homology of $\dist_N$, a task useful for
topological data analysis~\cite{edelsbrunner02topological}.
We also show how the nearest neighbor metric generalizes Euclidean
distance and maximum-edge Euclidean MST distance ~\cite{LiWan2001}

The core mathematical contribution of our work is the
statement and proof of Theorem~\ref{thm:NN}.  The techniques to prove
our other results are simpler and mostly leverage Theorem~\ref{thm:NN} and
past work. We have included them nonetheless to provide a more complete
picture of the nearest neighbor metric, and to provide possible directions for future
work.

%% Old outline: check this outline before submitting.
%% \begin{enumerate}
%% \item Definition of edge-squared.
%% \item Preliminaries
%% \item Outline/Overview/Previous Work
%% \item Interpretations based on known Machine Learning tools.
%% \begin{enumerate}
%% \item Gaussian Kernel similarity.
%% \item Generalization of Level-Set and Single-Linkage clustering. (\tim{This
%%     is not a very good point.})
%% \item Interpretation as $l_2$ on paths?
%% \end{enumerate}
%% \item Equality to a natural geodesic distance.
%% \begin{enumerate}
%% \item Define NN-metric.
%% \item Core Proof.
%% \item NN has a fast sparse spanner. Any spanner of the NN metric is a
%% spanner of edge-squared. (This result is theoretically
%%     superseded by our later result, but is of independent interest).
%% \item Persistent homology of the nearest neighbor metric can be
%% computed.
%% \end{enumerate}
%% \item Fast practical spanners for points in a distribution. ($k$-NN
%%     graph for $k = O( \log n)$.
%% \item Theoretically sparse, fast spanners for points in low dimension.  \tim{Low intrinsic dimension? Can cover trees do this for me?}
%% \item Remaining open questions.
%% \item Appendix:
%% \begin{enumerate}
%% \item Persistent homology can be put here, maybe.
%% \item Links to Heirarchies of Metrics: Edge-Squared is a natural
%% extension of negative type metrics, with potentially high doubling dimension.
%% \end{enumerate}
%%
%% \end{enumerate}

	% !TeX root = main.tex

\subsection{Definitions and Preliminaries} % (fold)
\label{sec:definitions}
In this section, we establish additional definitions for our paper. These
are mostly of interest for our spanner and persistent homology results, and are not strictly
necessary for Theorem~\ref{thm:NN}.

%  \tim{Anything on wireless networks here? Or other prelims?}
\vspace{3 mm}

\noindent \textbf{Spanners:} For real value $t \geq 1$, a $t$-spanner of
a weighted graph $G$ is a subgraph $S$ such that $d_G(x,y) \leq d_S(x,y)
\leq t\cdot d_G(x,y)$ where $d_G$ and $d_S$ represent the shortest path
distance functions between vertex pairs in $G$ and $S$. Spanners
of Euclidean distances, and general graph distances, have been
studied extensively, and their importance as a data structure is
well established.
~\cite{Chew1986, Vaidya1991, Callahan1993,HarPeled13}.

\vspace{3 mm}
\noindent \textbf{$k$-nearest neighbor graphs:} The $k$-nearest neighbor graph
($k$-NN graph) for a set of objects $V$ is a graph with vertex set $V$
and an edge from $v\in V$ to its $k$ most similar objects in $V$, under
a given distance measure. In this paper, the underlying distance
measure is Euclidean, and the edge weights are Euclidean distance
squared.
$k$-NN
graph constructions are a key data structure in machine
learning~\cite{Dong11, Chen11}, clustering~\cite{vL09}, and manifold learning~\cite{tenenbaum00global}.

\vspace{3 mm}
\noindent \textbf{Gabriel Graphs:} The Gabriel graph is a graph where
two vertices $p$ and $q$ are joined by an edge if and only if the disk
with diameter $pq$ has no other points of $S$ in the interior. The
Gabriel graph is a subgraph of the Delaunay
triangulation~\cite{SridharMaster}, and a
$1$-spanner of the edge-squared metric~\cite{SridharMaster}. Gabriel
graphs will be used in the proof of
Theorem~\ref{thm:distribution-spanner}.

\vspace{3 mm}
\noindent \textbf{Persistent Homology:}
  Persistent homology is a popular tool in computational geometry and topology to ascribe quantitative topological invariants to spaces that are stable with respect to perturbation of the input.
  In particular, it's possible to compare the so-called persistence diagram of a function defined on a sample to that of the complete space~\cite{chazal08towards}.
  These two aspects of persistence theory---the intrinsic nature of topological invariants and the ability to rigorously compare the discrete and the continuous---are both also present in our theory of nearest neighbor metrics.
  Indeed, our primary motivation for studying these metrics was to use them as inputs to persistence computations for problems such as persistence-based clustering~\cite{chazal13persistence} or metric graph reconstruction~\cite{aanjaneya12metric}.

% section definitions (end)

  \section{Outline}
%   Definitions of the nearest neighbor metric, and of the
%     edge-squared metric, are provided in
%    Section~\ref{sec:definitions}.
%
Section~\ref{sec:NN} contains the proof of Theorem~\ref{thm:NN},
equating the edge-squared metric and nearest neighbor metric in
all cases. It should be noted that our proof is robust enough to handle not
just finite point sets, but also countably infinite collections of disjoint
path-connected, compact sets. Remarkably, there is no restriction on the convexity or
simply-connectedness of these sets.

As an example of using the nearest neighbor metric to compute intrinsic structure, Section~\ref{sec:persistence} shows how Theorem~\ref{thm:NN} allows us to compute the persistent homology of the nearest neighbor metric.

Section~\ref{sec:edge-power} introduces the $p$-power metrics. We show
that Euclidean spanners and Euclidean MSTs are special cases of
$p$-power spanners. We show how
clustering algorithms including $k$-means, level-set methods,
and single linkage clustering, are special cases of
clustering with $p$-power metrics. $p$-power metrics are identical to the
Neighbor metric when $p=2$. This is further detailed in
Appendix~\ref{ap:clustering-link}.

Section~\ref{sec:general-spanner} outlines a proof of
Theorem~\ref{thm:general-spanner}, and compares our spanner to new lower
bounds on the sparsity of $(1+\eps)$-spanners of the Euclidean metric.  We
outline a proof of Theorem~\ref{thm:distribution-spanner} in
Section~\ref{sec:distribution-spanner} and discuss its implications.

Conclusions and open questions are in
Section~\ref{sec:conclusions}. Full proofs for
Theorems~\ref{thm:distribution-spanner},~\ref{thm:general-spanner}
are contained in the Appendix.

  \section{Exactly Computing the nearest neighbor metric}
\label{sec:NN}
In this section, we prove Theorem~\ref{thm:NN} on finite point sets, and
explain in Section~\ref{sec:bodies} that our proof strategy applies to
finite collections of path-connected compact bodies.

  First, lets observe what happens when $P$ has only two points $a$ and
$b$, $\dist_2(a,b) = \dist_N(a,b)$.  This reduces to a high school calculus
exercise as the minimum path $\gamma$ will be a straight line between the
points and the nearest neighbor metric is
\begin{align*}
  \dist_N(a,b) &= 4\int_0^1 \distto_P(\gamma(t))\|\gamma'(t)\|dt \\
  & = 8\int_0^{\frac{1}{2}} t \|a - b\|^2 dt = \|a - b\|^2 = \dist_2(a,b).
\end{align*}
  Now it is easy to observe that the
  nearest neighbor metric is never greater than the
  edge-squared distance, as proven in the following lemma.

  \begin{lemma}\label{lem:dist_N_le_dist}
    For all $s,p\in P$, we have $\dist_N(s,p)\le \dist_2(s,p)$.
  \end{lemma}
  \begin{proof}
    Fix any points $s,p\in P$.
    Let $q_0,\ldots, q_k \in P$ be such that $q_0 = s$, $q_k = p$ and
    \[
      \dist_2(s,p) = \sum_{i=1}^k \|q_i - q_{i-1}\|^2.
    \]
    Let $\psi_i(t) = tq_i + (1-t)q_{i-1}$ be the straight line segment from $q_{i-1}$ to $q_i$.
    Observe that $\len(\psi_i) = \|q_i - q_{i-1}\|^2 / 4$, by the same argument as in the two point case.
    Then, let $\psi$ be the concatenation of the $\psi_i$ and it follows that
    \[
      \dist_2(s,p) = 4 \len(\psi) \ge 4 \inf_{\gamma\in \ourpath(s,p)} \len(\gamma) = \dist_N(s,p).
    \]
  \end{proof}

  By Lemma~\ref{lem:dist_N_le_dist}, it suffices to show that $\dist_N(a, b) \geq \dist_2(a,b)$ for all $a, b \in P$.
% This allows us to compute of the nearest neighbor metric exactly by instead computing the edge-squared metric.
%\subsection{Equivalence} % (fold)
\label{sec:the_proof}

  Let $P\subset \R^d$ be a set of $n$ points.
  Pick any \emph{source} point $s\in P$.
  Order the points of $P$ as $p_1,\ldots ,p_n$ so that
  \[
    \dist_2(s,p_1) \le \cdots \le \dist_2(s, p_n).
  \]
  This will imply that $p_1 = s$.
  It will suffice to show that for all $p_i\in P$, we have $\dist_2(s,p_i) = \dist_N(s,p_i)$.
  There are three main steps:
  \begin{enumerate}
    \item We first show that when $P$ is a subset of the vertices of an axis-aligned box, $\dist = \dist_N$.  In this case, shortest paths for $\dist$ are single edges and shortest paths for $\dist_N$ are straight lines.
    \item We then show how to lift the points from $\R^d$ to $\R^n$ by a Lipschitz map $m$ that places all the points on the vertices of a box and preserves $\dist_2(s,p)$ for all $p\in P$.
    \item Finally, we show how the Lipschitz extension of $m$ is also Lipschitz as a function between nearest neighbor metrics.  We combine these pieces to show that $\dist \le \dist_N$.  As $\dist \ge \dist_N$ (Lemma~\ref{lem:dist_N_le_dist}), this will conclude the proof that $\dist = \dist_N$.
  \end{enumerate}
The key to the second step, to be elaborated in Section~\ref{sec:lifting},
is that if you take points on a line and raise the pairwise distances to
the $1/2$ power, you get points on a box. This is a special case of the
general theory on screw functions developed by Von Neumann and Schoenberg,
which asserts a far more general criterion on when functions applied to
pairwise distances between points on a line can be embedded into Euclidean
space~\cite{VonNeumann41}.
  % !TeX root = main.tex

\subsubsection{Boxes} % (fold)
\label{sec:boxes}

  Let $Q$ be the vertices of a box in $\R^n$.
  That is, there exist some positive real numbers $\alpha_1,\ldots , \alpha_n$ such that each $q\in Q$ can be written as $q = \sum_{i\in I} \alpha_i e_i$, for some $I\subseteq [n]$.

  Let the source $s$ be the origin.
  Let $\distto_Q:\R^n\to \R$ be the distance function to the set $Q$.
  Setting $r_i(x) := \min\{x_i, \alpha_i - x_i\}$ (a lower bound on the difference in the $i$th coordinate to a vertex of the box), it follows that
  \begin{equation}
    \label{eq:distto_bded_by_ris}
    \distto_Q(x) \ge \sqrt{\sum_{i= 1}^n r_i(x)^2}.
  \end{equation}

  Let $\gamma:[0,1]\to \R^n$ be a curve in $\R^n$.
  Define $\gamma_i(t)$ to be the projection of $\gamma$ onto its $i$th coordinate.
  Thus,
  \begin{equation}\label{eq:ri_as_min}
    r_i(\gamma(t)) = \min\{\gamma_i(t), \alpha_i - \gamma_i(t)\}
  \end{equation}
  and
  \begin{equation}\label{eq:gamma_decomposed}
    \|\gamma'(t)\| = \sqrt{\sum_{i = 1}^n \gamma_i'(t)^2}.
  \end{equation}
  We can bound the length of $\gamma$ as follows. For simplicity of exposition we only present the case
  of a path from the origin to the far corner, $p = \sum_{i=1}^n \alpha_i e_i$. 
  \begin{align*}
    \len(\gamma)
      &= \int_0^1 \distto_Q(\gamma(t))\|\gamma'(t)\|dt \\
      & \text{[by definition]}\\
      &\ge \int_0^1 \left(\sqrt{\sum_{i= 1}^n r_i(\gamma(t))^2} \sqrt{\sum_{i = 1}^n \gamma_i'(t)^2}\right) dt \\
      & \text{[by \eqref{eq:distto_bded_by_ris} and \eqref{eq:gamma_decomposed}]}\\
      &\ge \sum_{i=1}^n \int_0^1 r_i(\gamma(t)) \gamma_i'(t) dt \\
      & \text{[by Cauchy-Schwarz]}\\
      &\ge \sum_{i=1}^n \left(\int_0^{\ell_i} \gamma_i(t) \gamma_i'(t)dt + \int_{\ell_i'}^1 (\alpha_i - \gamma_i(t)) \gamma_i'(t) dt\right) \\
      & \text{[by \eqref{eq:ri_as_min} where $\gamma_i(\ell_i) = \alpha_i/2$ for the first time} \\
      & \text{and $\gamma_i(\ell_i') = \alpha_i/2$ for the last time.]}\\
      &= \sum_{i=1}^n 2\int_0^{\ell_i} \gamma_i(t) \gamma_i'(t) dt \\
      & \text{[by symmetry]}\\
      &\ge \sum_{i=1}^n \frac{\alpha_i^2}{4} \\
      & \text{[by basic calculus]}
  \end{align*}

  It follows that if $\gamma$ is any curve that starts at $s$ and ends at $p = \sum_{i=1}^n \alpha_i e_i$, then $\dist_N(s,p) = \dist_2(s,p)$.

% section boxes (end)

  \subsubsection{Lifting the points to $\R^n$} % (fold)
\label{sec:lifting}

  Define a mapping $m: P \to \R^n$.  We do this by adding the points $p_1, \ldots, p_n$, as defined above, one point at a time.
  For each new point we will introduce a new dimension. We start by setting $m(p_1) = 0$ and by induction:
  \begin{equation}\label{eq:defn_of_m}
    m(p_i) = m(p_{i-1}) + \sqrt{\dist_2(s, p_i) - \dist_2(s, p_{i-1})} e_i,
  \end{equation}
  where the vectors $e_i$ are the standard basis vectors in $\R^n$.
  A similar embedding works for some other functions and was extensively studied by Schoenberg and Von Neumann in the theory of screw functions.

  \begin{lemma}\label{lem:m_and_dist}
    For all $p_i, p_j\in P$, we have
    \begin{enumerate}
      \item[(i)] $\|m(p_j) - m(p_i)\| = \sqrt{|\dist_2(s,p_j) - \dist_2(s,p_i)|}$, and
      \item[(ii)]$\|m(s) - m(p_j)\|^2 \le \|m(p_i)\|^2 + \|m(p_i) - m(p_j)\|^2$.
    \end{enumerate}
  \end{lemma}
  \begin{proof}
    \emph{Proof of (i).}
    Without loss of generality, let $i \le j$.
    Then, by the definition of $m$, expanding the norm, and telescoping the sum, we get the following.
    \begin{align*}
      & \|m(p_j) - m(p_i)\| \\
      &= \left\|\sum_{k=i+1}^j \sqrt{\dist_2(s, p_k) - \dist_2(s, p_{k-1})} e_k \right\| \\
      &= \sqrt{\sum_{k=i+1}^j (\dist_2(s, p_k) - \dist_2(s, p_{k-1}))}\\
      &= \sqrt{\dist_2(s, p_j) - \dist_2(s, p_i)}.
    \end{align*}

    \noindent\emph{Proof of (ii).}
    As $m(s) = 0$, it suffice to observe that
    \begin{align*}
      \|m(p_j)\|^2
        &= \dist_2(s, p_j) \because{by \emph{(i)}}\\
        % & \text{[by \emph{(i)}]}\\
        &\le \dist_2(s, p_i) + |\dist_2(s,p_j) - \dist_2(s,p_i)| \\
        % & \text{[by basic arithmetic]}\\
        &= \|m_(p_i)\|^2 + \|m(p_i) - m(p_j)\|^2  \because{by \emph{(i)}}
        % & \text{[by \emph{(i)}]}
    \end{align*}
  \end{proof}

  We can now show that $m$ has all of the desired properties.

  \begin{prop}\label{prop:m_is_good}
    Let $P\subset\R^d$ be a set of $n$ points, let $s\in P$ be a designated source point, and let $m:P\to \R^n$ be the map defined as in \eqref{eq:defn_of_m}.
    Let $\dist'$ denote the edge squared metric for the point set $m(P)$ in $\R^n$.
    Then,
    \begin{enumerate}
      \item[(i)] $m$ is $1$-Lipschitz as a map between Euclidean metrics,
      \item[(ii)] $m$ maps the points of $P$ to the vertices of a box, and
      \item[(iii)] $m$ preserves the edge squared distance to $s$, i.e.\ $\dist'(m(s), m(p)) = \dist_2(s,p)$ for all $p\in P$.
    \end{enumerate}
  \end{prop}
  \begin{proof}
    \emph{Proof of (i).} To prove the Lipschitz condition, fix any $a,b\in P$ and bound the distance as follows.
    \begin{align*}
      \|m(a) - m(b)\|
        &= \sqrt{|\dist_2(s,a) - \dist_2(s,b)|} \because{Lem.~\ref{lem:m_and_dist}(i)}\\
        % & \text{[by Lemma~\ref{lem:m_and_dist}(i)]}\\
        &\le \sqrt{\dist_2(a,b)} \because{triangle ineq.}\\
        % & \text{[by triangle inequality]}\\
        &\le \|a-b\|. \because{by def. of $\dist_2$}\\
        % & \text{[$\dist_2(a,b)\le \|a-b\|^2$ by the definition of $\dist$]}
    \end{align*}

    \noindent
    \emph{Proof of (ii).} That $m$ maps $P$ to the vertices of a box is immediate from the definition.
    The box has side lengths $\|m_i - m_{i-1}\|$ for all $i>1$ and $p_i = \sum_{k=1}^i \|m_k - m_{k-1}\| e_k$.

    \noindent
    \emph{Proof of (iii).} We can now show that the edge squared distance to $s$ is preserved.
    Let $q_0,\ldots, q_k$ be the shortest sequence of points of $m(P)$ that realizes the edge-squared distance from $m(s)$ to $m(p)$, i.e., $q_0 = m(s)$, $q_k = m(p)$, and
    \[
      \dist'(m(s), m(p)) = \sum_{i = 1}^k \|m(q_i) - m(q_{i-1})\|^2.
    \]
    If $k> 1$, then Lemma~\ref{lem:m_and_dist}(ii) implies that removing $q_1$ gives a shorter sequence.
    Thus, we may assume $k = 1$ and therefore, by Lemma~\ref{lem:m_and_dist}(i),
    \[
      \dist'(m(s), m(p)) = \|m(s) - m(p)\|^2 = \dist_2(s,p). %\qedhere
    \]
  \end{proof}

% section lifting (end)

  \subsubsection{The Lipschitz Extension} % (fold)
\label{sec:lip_extension}

  Proposition~\ref{prop:m_is_good} and the Kirszbraun theorem on Lipschitz extensions imply that we can extend $m$ to a $1$-Lipschitz function $f: \R^d\to \R^n$ such that $f(p) = m(p)$ for all $p\in P$ \cite{Kirszbraun1934,Valentine1945,brehm1981}.

  \begin{lemma}\label{lem:dist_N_lipschitz}
    The function $f$ is also $1$-Lipschitz as mapping from $\R^d\to \R^n$ with both spaces endowed with the nearest neighbor metric.
  \end{lemma}
  \begin{proof}
    We are interested in two distance functions $\distto_P:\R^d \to \R$ and $\distto_{f(P)}: \R^n\to \R$.
    Recall that each is the distance to the nearest point in $P$ or $f(P)$ respectively.
    \begin{align*}
      \distto_{f(P)}(f(x))
        &= \min_{q\in f(P)} \|q - f(x)\| \because{by definition}\\
        % & \text{[by definition]}\\
        &= \min_{p\in P} \|f(p) - f(x)\| \because{$q\in f(P)$}\\
        % & \text{[$q = f(p)$ for some $p$]}\\
        &\le \min_{p\in P} \|p - x\|\because{$f$ is $1$-Lipschitz}\\
        % & \text{[$f$ is $1$-Lipschitz]}\\
        &= \distto_P(x). \because{by definition}\\
        % & \text{[by definition]}
    \end{align*}
    For any curve $\gamma:[0,1]\to \R^d$ and for all $t\in [0,1]$, we have $\|(f\circ \gamma)'(t)\| \le \|\gamma'(t)\|$.
    It then follows that

    \begin{align}
      \len'(f\circ \gamma) &= \int_0^1 \distto_{f(P)}(f(\gamma(t)))\|(f\circ\gamma)'(t)\|dt \nonumber\\
              & \le \int_0^1 \distto_{P}(\gamma(t))\|\gamma'(t)\|dt = \len(\gamma),\label{eq:curves_shorten}
      \end{align}
    where $\len'$ denotes the length with respect to $\distto_{f(P)}$.
    Thus, for all $a,b\in P$,
    \begin{align*}
      \dist_N(a,b)
        &= 4 \inf_{\gamma\in \ourpath(a,b)} \len(\gamma) \because{by definition}\\
        &\ge 4 \inf_{\gamma\in \ourpath(a,b)} \len'(f\circ\gamma) \because{by \eqref{eq:curves_shorten}}\\
        % & \text{[by \eqref{eq:curves_shorten}]}\\
        &\ge 4 \inf_{\gamma'\in \ourpath(f(a),f(b))} \len'(\gamma') \because{$f\circ\gamma$ is a path}\\
        % & \text{[because $f\circ\gamma\in \ourpath(f(a), f(b))$]}\\
        &= \dist_N(f(a), f(b)). \because{by definition}\\
        % & \text{[by definition]}
    \end{align*}
  \end{proof}

  We now restate Theorem~\ref{thm:NN} for convenience, and prove
  it.
  \begin{theorem}\label{thm:equality}
    For any point set $P\subset\R^d$, the edge squared metric $\dist$ and the nearest neighbor metric $\dist_N$ are identical.
  \end{theorem}
  \begin{proof}
    Fix any pair of points $s$ and $p$ in $P$.
    Define the Lipschitz mapping $m$ and its extension $f$ as in \eqref{eq:defn_of_m}.
    Let $\dist'$ and $\dist_N'$ denote the edge-squared and nearest neighbor metrics on $f(P)$ in $\R^n$.
    \begin{align*}
      \dist_2(s,p)
        &= \dist'(m(s), m(p)) \because{Proposition~\ref{prop:m_is_good}(iii)}\\
        &= \dist_N'(m(s), m(p)) \because{$f(P)$ are vertices of a box}\\
        &\le \dist_N(s, p) \because{Lemma~\ref{lem:dist_N_lipschitz}}
    \end{align*}
    We have just shown that $\dist\le \dist_N$ and Lemma~\ref{lem:dist_N_le_dist} states that $\dist\ge \dist_N$, so we conclude that $\dist = \dist_N$ as desired.
  \end{proof}

% section lip_extension (end)

\subsection{From Finite Sets to Finite Collections of Compact Path-Connected Bodies}
\label{sec:bodies}
All of our proof steps hold for finite collections of compact,
path-connected bodies in arbitrarily large dimension. Our Lipschitz map $m$ can
still be extended to a Lipschitz map $f$ in this setting, largely due to the
generality of the Kirszbraun theorem. In this case, the pre-image of the
contractive map is the set of all points belonging to some body. Meanwhile, the image is a finite set of points, the corners of a multi-dimensional box.
Thus our construction of $m$ contracts each convex
body into a single point, and the image of our
compact bodies under $f$ is still a finite point set on the corners of a
box. Therefore, the remainder of our theorem proof goes through unchanged.

This result is rather remarkable: path-connected compact sets in high
dimensional space can have extremely convoluted geometry, and the Voronoi diagrams
on these collections (on which the nearest neighbor metric depends) can be
massively complex.  The key is that our Lipschitz map is robust enough to
handle objects of considerable geometric complexity.

% section the_proof (end)

  \section{Persistent Homology of the Nearest-neighbor Geodesic Distance}
  \label{sec:persistence}

% \tim{This section should be prefaced somewhere -- here or in a previous
% section -- with a statement as to why the nearest neighbor metric may
% still be of independent interest. My main claim is that it's useful
% since it's defined on all points rather than just two.}
%
% \tim{My main issue with this section is that I don't have a clean
% theorem saying how to compute persistent homology, both in ambient and
% intrinsic setting. Having one or two top-level theorems that say this
% would be great, rather than having it be written in the exposition. As
% it stands, I have no clue how Lemma 4.5 is used, or Lemma 4.6, to
% compute the homologies.}
  In this section, we show how to compute the so-called persistent homology~\cite{edelsbrunner02topological} of the nearest neighbor metric in two different ways, one ambient and the other intrinsic.
  The latter relies on Theorem~\ref{thm:NN} and would be quite surprising without it.

  The input for persistence computation is a \emph{filtration}---a nested sequence of spaces, usually parameterized by a real number $\alpha\ge 0$.
  The output is a set of points in the plane called a \emph{persistence diagram} that encodes the birth and death of topological features like connected components, holes, and voids.

  \paragraph*{The Ambient Persistent Homology}

    Perhaps the most popular filtration to consider on a Euclidean space is the sublevel set filtration of the distance to a sample $P$.
    This filtration is $(F_\alpha)_{\alpha\ge 0}$, where
    \[
      F_\alpha := \{x\in \R^d \mid \distto_P(x) \le \alpha\},
    \]
    for all $\alpha \ge 0$.
    If one wanted to consider instead the nearest neighbor metric $\dist_N$, one gets instead a filtration $(G_\alpha)_{\alpha\ge 0}$, where
    \[
      G_\alpha := \{x\in \R^d \mid \min_{p \in P} \dist_N(x,p) \le \alpha\},
    \]
    for all $\alpha \ge 0$.

    Both the filtrations $(F_\alpha)$ and $(G_\alpha)$ are unions of metric balls.
    In the former, they are Euclidean.
    In the latter, they are the metric balls of $\dist_N$.
    These balls can look very different, for example, for $\dist_N$, the metric balls are likely not even convex.
    However, these filtrations are very closely related.
    \begin{lemma}
      For all $\alpha \ge 0$, $F_\alpha = G_{2\alpha^2}$.
    \end{lemma}
    \begin{proof}
      The key to this exercise is to observe that the nearest point $p \in P$ to a point $x$ is also the point that minimizes $\dist_N(x,p)$.
      To prove this, we will show that for any $p\in P$ and any path $\gamma\in \ourpath(x,p)$, we have $\len(\gamma)\ge \frac{1}{2}\distto_P(x)^2$.
      Consider any such $x$, $p$, and $\gamma$.
      The euclidean length of $\gamma$ must be at least $\distto_P(x)$, so we will assume that $\|\gamma'\| = \distto_P(x)$ and will prove the lower bound on the subpath starting at $x$ of length exactly $\distto_P(x)$.
      This will imply a lower bound on the whole path.
      Because $\distto_P$ is $1$-Lipschitz, we have $\distto_P(\gamma(t))\ge (1-t)\distto_P(x)$ for all $t\in [0,1]$.  It follows that
      \begin{align*}
        \len(\gamma) &= \int_0^1 \distto_P(\gamma(t))\|\gamma'(t)\|dt \\
         &  \ge \distto_P(x)^2 \int_0^1 (1-t)dt
          =  \frac{1}{2}\distto_P(x)^2
      \end{align*}
      The bound above applies to any path from $x$ to a point $p\in P$, and so,
      \[
        \dist_N(x,p) = 4 \inf_{\gamma\in \ourpath(x,p)} \len(\gamma) \ge 2\distto_P(x).
      \]
      If $p$ is the nearest neighbor of $x$ in $P$, then $\dist_N(x,p) = 2\distto_P(x)$, by taking the path to be a straight line.
      It follows that $\min_{p\in P}\dist_N(x,p) = 2\distto_P(x)$.
    \end{proof}

    The preceding lemma shows that the two filtrations are equal up to a monotone change in parameters.
    By standard results in persistent homology, this means that their persistence diagrams are also equal up to the same change in parameters.
    This means that one could use standard techniques such as $\alpha$-complexes~\cite{edelsbrunner02topological} to compute the persistence diagram of the Euclidean distance and convert it to the nearest neighbor metric afterwards.
    Moreover, one observes that the same equivalence will hold for variants of the nearest neighbor metric that take other powers of the distance.

    % (some are considered in Section~\ref{sec:higher_powers}).

  \paragraph*{Intrinsic Persistent Homology}
    Recently, several researchers have considered intrinsic nerve complexes on metric data, especially data coming from metric graphs~\cite{adamszek16nerve,gasparovic17complete}.
    These complexes are defined in terms of the intersections of metric balls in the input.
    The vertex set is the input point set.
    The edges at scale $\alpha$ are pairs of points whose $\alpha$-radius balls intersect.
    In the intrinsic \v Cech complex, triangles are defined for three way intersections, tetrahedra for four-way intersections, etc.

    In Euclidean settings, little attention was given to the difference between the intrinsic and the ambient persistence, because a classic result, the Nerve Theorem~\cite{borsuk48imbedding}, and its persistent version~\cite{chazal08towards} guaranteed there is no difference.
    The Nerve theorem, however, requires the common intersections to be contractible, a property easily satisfied by convex sets such as Euclidean balls.
    However, in many other topological metric spaces, the metric
    balls might not be so well-behaved.
    In particular, the nearest neighbor metric has metric balls which may take on very strange shapes, depending on the density of the sample.
    This is similarly true for graph metrics.
    So, in these cases, there is a difference between the information in the ambient and the intrinsic persistent homology.

    \begin{theorem}
      Let $P\subset \R^d$ be finite and let $\dist_N$ be the nearest neighbor metric with respect to $P$.
      The edges of the intrinsic \v Cech filtration with respect to $\dist_N$ can be computed exactly in polynomial time.
    \end{theorem}
    \begin{proof}
      The statement is equivalent to the claim that $\dist_N$ can be computed exactly between pairs of points of $P$, a corollary of Theorem~\ref{thm:equality}.
      Two radius $\alpha$ balls will intersect if and only of the distance between their centers is at most $2\alpha$.
      The bound on the distance necessarily implies a path and the common intersection will be the midpoint of the path.
    \end{proof}

  \section{Relating the nearest neighbor metric to Euclidean MSTs,
Euclidean Spanners, and More}\label{sec:edge-power}
The nearest neighbor metric, as seen in Theorem~\ref{thm:NN}, is equal to
the edge-squared metric. This allows us to connect this manifold distance
to a graph distance, which we will in turn show is a generalization of
maximum-edge distance on minimum spanning trees. The results in this
section are quite simple to prove, but we nonetheless believe they are important
properties of the Nearest Neighbor metric and its variants.

The edge-squared metric on a Euclidean point set, as we recall, is defined by taking the
Euclidean distances squared and finding the shortest paths. 
We could have taken any such power $p$ of the Euclidean distances. We
will soon see that taking $p = 1$ gives us the Euclidean distance, and
finding spanners of the graph as $\lim p\rightarrow \infty$ is the
Euclidean MST problem.  
Let the $p$-power metric be defined on a Euclidean point set by taking
Euclidean distances to the power of $p$, and performing all-pairs
shortest path on the resulting distance graph.
\begin{theorem} \label{thm:1-spanner}
For all $q > p$, any $1$-spanner of the $p$-power metric is a
$1$-spanner of the $q$-power metric on the same point set
\end{theorem}
\begin{proof}
A $1$-spanner of the $q$-power metric can be made by taking
  edges $uv$ where
  \begin{align}\label{eq:span}
    \min_{p_0=u, \ldots p_k=v, k \not= 1} \sum_k ||p_i - p_{i-1}||^q > || u -
v||^q.
  \end{align}
If
  $\sum_{i=1}^k ||p_i - p_{i-1}||^q > || u - v||^q$ for any points
  $p_1, \ldots p_k$, then
  $\sum_{i=1}^k ||p_i - p_{i-1}||^p > || u - v||^p$ for any $q > p$.
  Thus, for all such edges $uv$ satisfying
  Equation~\ref{eq:span}:
  \[ \min_{p_0=u, \ldots p_k=v, k \not= 1} \sum_k ||p_i - p_{i-1}||^p > || u -
  v||^p. \] Such edges $uv$ must be included in any $1$-spanner
  of the $p$-power
  metric.
\end{proof}

 \begin{corollary}
   Let $P$ be a set of points in Euclidean space drawn i.i.d. from a Lipschitz
   probability density bounded above and below, with support on a
   smooth, compact manifold with intrinsic dimension $d$, bounded
   curvature, and
   smooth boundary of bounded curvature. Then the $k$-NN graph on $P$
   when $k = O(2^d \log n)$ is a $1$-spanner of the $p$-power
   metric for every $p \geq 2$, w.h.p.
 \end{corollary}
This follows from combining Theorem~\ref{thm:distribution-spanner} and
Theorem~\ref{thm:1-spanner}.
\subsection{Relation to the Euclidean MST problem}
  \begin{definition}
  Let the \textbf{normalized $p$-power metric} between two points in
  $\mathbb{R}^d$ be the $p$-power metric between the two points,
  raised to the $\frac{1}{p}$ power. Define the normalized $\infty$-power
  metric as the limit of the normalized $p$-power metric as $p \rightarrow \infty$.
  \end{definition}
  \begin{lemma} The Euclidean MST is a
  $1$-spanner for the normalized $\infty$-power metric.
  \end{lemma}
  This lemma follows from basic properties of the MST.
  The normalized $p$-power metrics give us a suite of
  metrics such that $p=1$ is the Euclidean
  distance and $p=\infty$ gives us the distance of the longest edge on the
  unique MST-path.  Setting $p=2$ gives the edge-squared metric, which
  sits between the Euclidean and max-edge-on-MST-path distance. 
  Theorem~\ref{thm:1-spanner} establishes
  that minimal $1$-spanners of the (normalized) $p$-power
  metric are contained in each other, as $p$ varies from
  $1$ to $\infty$. The minimal spanner for a general point set when $p=1$ is the complete graph, and
 the Euclidean MST is the minimal spanner for $p=\infty$. Thus:
  \begin{theorem} 
    For points in $\mathbb{R}^d$, every $1$-spanner of the $p$-power
    metric on that set of points contains every Euclidean MST.
  \end{theorem}
  \begin{corollary}
    Every $1$-spanner for the Nearest Neighbor metric
    contains every Euclidean
    MST. 
  \end{corollary}
  \subsection{Generalizing Single Linkage Clustering, Level Sets, and k-Centers
  clustering}
  If our point set is drawn from a well-behaved probability
  density, then the normalized edge-power metrics
  converge to a nice geodesic distance detailed
  in~\cite{hwang2016}. When $p=1$, clustering with this metric is
  the same as Euclidean metric clustering ($k$-means,
  $k$-medians, $k$-centers), and when $p=\infty$, clustering with
  this metric is the same as the single-linkage clustering and the widely
used level-set method~\cite{Wishart69, Gower1969, Ester1996,OPTICS96}.
Thus, clustering with normalized edge-power metrics generalizes these two
very popular methods, and interpolates between their advantages.
Definitions of the level-set method and a full discussion are contained in
Appendix~\ref{ap:clustering-link}

% \tim{Maybe you want to mention the self-sparsifying graph when you raise
% $p$ from $1$ to $\infty$}

% Kicked edge-power to appendix
  \section{Spanners for the nearest neighbor metric}
In this section, we prove our theorems on spanners of the nearest neighbor
metric. The proofs of these theorems mostly leverage Theorem~\ref{thm:NN}
and past work on geometric spanners. We have nonetheless included them
for completeness, and to illustrate that spanners of
manifold distances like the nearest neighbor metric can have interesting
properties not found in Euclidean spanners (assuming no Steiner points).
\subsection{Exact-spanners of nearest neighbor metric in
the Probability Density Setting
}\label{sec:distribution-spanner}
% \tim{Do a warm up on Gabriel graphs with convex support.}

Theorem~\ref{thm:distribution-spanner} states that for $k= O(2^d \log n)$, the $k$-NN graph of $n$ points drawn i.i.d from
a nicely behaved probability distribution is a $1$-spanner of the
nearest neighbor metric. This section is dedicated to outlining a proof of
this Theorem, the full result which will be in
Appendix~\ref{ap:distribution-spanner}.
This result is clearly impossible for Euclidean distances,
whose $1$-spanner is the complete graph almost surely.  
Our theorem implies any off-the-shelf
$k$-nearest neighbor graph generator can compute
edge-squared metric.  We strongly rely on Theorem~\ref{thm:NN} for this
result, and the fact that Gabriel graphs are $1$-spanners of the
edge-squared metric.

First, let us assume that the support of our
probability density $D$ has the same dimension as our ambient space.
This simplifies our calculations without changing the problem
much. Then, we note that as our number of sample points get
large, the density inside a $k$-NN ball around any point $x$ (the ball with radius
$k^{th}$-NN
distance, center at $x$) looks like the uniform distribution
on that ball, possibly intersected with a halfspace. The bounding
plane of our halfspace
represents the boundary of our density $D$.

For simplicity in the outline, let's suppose that $D$ is convex.
If we condition on
the radius of the $k$-NN ball, then the $k-1^{st}$ nearest
neighbors of $x$ are distributed roughly according to the above
distribution, described by the ball intersected with a halfspace.
For any other point $p$ in $D$, we project $p$ onto the $k$-NN
ball to point $p'$, and show that the ball $p'x$ contains a $k^{th}$ nearest
neighbor w.h.p, when $k=O(2^d \log n)$. This
implies ball with diameter $px$ contains a $k^{th}$ nearest
neighbor of $x$, and thus $px$ is not necessary in any
$1$-spanner of the edge-squared metric. Then we take union bound
over all $x$.
A rigorous proof of
Theorem~\ref{thm:distribution-spanner} requires careful analysis, 
and is contained in
Section~\ref{ap:distribution-spanner}.  
Our proof can be tweaked to show:
\begin{theorem}
  Given a Lipschitz distribution bounded above and below with support on convex set $C \subset
  \mathbb{R}^d$, the $k$-NN graph is Gabriel w.h.p. for $k =
  O(2^d \log n)$.
\end{theorem}

\subsection{Fast, Sparse Spanner for the Edge-Squared
Metric}
\label{sec:general-spanner}

Now we outline a proof for Theorem~\ref{thm:general-spanner}, which shows
that one can construct a $(1+\eps)$ nearest neighbor metric spanner of size
$O(n \eps^{-d/2})$ in
time $O\left(n \log n + n \eps^{-d/2} \log \left(\frac{1}{\eps}\right)\right)$, 
for points in constant dimensional space. The full proof is in
Appendix~\ref{ap:general-spanner}. We critically rely on
Theorem~\ref{thm:NN} for this work, which shows
a spanner for the edge-squared metric is equivalent to a spanner for the
nearest neighbor metric.

Note that this spanner is sparser and faster in terms of epsilon dependency
than the theoretical optimal spanner for Euclidean distances~\cite{Le19}.
We rely extensively on well-separated pair decompositions
(WSPDs),
and this outline assumes familiarity with that notation.
For a comprehensive set of definitions and notations on well separated
pairs, refer to any of \cite{Callahan1995, Arya2016, Callahan1993,
arya95euclid}.  
Our proof consists of three parts.
\begin{enumerate}
\item Showing that connecting a $(1+O(\delta^2))$-approximate shortest edge
in a $1/\delta$ well separated pair for all the pairs in the decomposition
gives a $1+O(\delta^2)$ edge-squared spanner.
The processing for this step takes $O(n \log n + \delta^{-d}n)$ time.
\item Previous work contains an algorithm computing
  $1+O(\delta^2)$-approximate shortest edge in a $1/\delta$ well
    separated pair for all the pairs in a WSPD, and takes
    $O(1)$ time per pair. The pre-processing for this step will be
    bounded by $O(\delta^{-d}n\log\left(\frac{1}{\delta}\right))$ time. The $\log\left(\frac{1}{\delta}\right)$ factor goes away given a fast floor function. 
    This procedure was first introduced in~\cite{Callahan1995}.

\item Putting these two together, and setting $\epsilon = \delta^2$
gives us a $1+\epsilon$ spanner with
$O(\epsilon^{-d/2}n)$ edges in
    $O(n \log n + \epsilon^{-d/2}n)$ time.
\end{enumerate}
Full details of this proof are contained in Appendix~\ref{ap:general-spanner}

  \section{Conclusions and Open Questions}\label{sec:conclusions}

We examined the nearest neighbor metric and showed how to compute it
exactly, as well as find sparse data structures efficiently for approximate
computation in practice.  Many problems remain open.

First: are there generalizations of these metrics, for which our proof techniques
will still hold? The nearest neighbor metric has many natural
generalizations, including the $k^{th}$ nearest neighbor or powers of the
nearest neighbor function.

Can we efficiently compute $o(\log n)$-spanners of the nearest neighbor
metric in high dimension, such the the spanners have a nearly linear number of edges?
The existence of such spanners has been studied for Euclidean metrics in~\cite{HarPeled13}, where the stretch
obtained is $\sqrt{\log n}$.

Does computing $k$-NN graphs with approximate nearest neighbor methods give $1$-spanners of the
edge-squared metric with high probability?
Approximate nearest neighbors have been studied extensively
~\cite{kNNsurvey, Chen11, Dong11}, including locality-sensitive
hashing for high dimensional point sets~\cite{LSH} and
more~\cite{Laarhoven2018}.
Recent work by Andoni et al.~\cite{Andoni2018} showed how to compute
approximate nearest neighbors for any non-Euclidean norm.  Perhaps there
is a rigorous theory about data-sensitive metrics generated from any such
norm? Similar to
how the edge-squared metric is generated from the Euclidean distance.

It remains an open question how well clustering or
classification with
nearest neighbor metrics performs on real-world
data.
Experiments have been done by Bijral, Ratliff, and
Srebro in~\cite{bijral11semiSupLearningDBD}.
Theorem~\ref{thm:distribution-spanner} implies that future experiments
can be done using any
k-nearest-neighbor graph.  
  We believe that the interest in alternative metrics on Euclidean data will continue to be a rich source of interesting problems.

%%   We showed that the class of metrics that can be represented in this way is quite large, including all negative-type metrics and more.
%%   There is no constant upper bound on the doubling dimension of such metrics even for point sets in the plane.
%%   Theorem~\ref{thm:equality} extends these results immediately to density-based distances.
%%   Although there is such rich structure in these metrics, they still maintain some of the advantages of low-dimensional euclidean metrics, and this is exploited to generate sparse spanners.
%%   We also use the Euclidean structure to compute the persistent homology of the nearest neighbor distance in both the ambient sense and the intrinsic sense.
%%
%%   Many interesting problems remain open.
%%   One of the most compelling is the generalization of Theorem~\ref{thm:equality} to higher power metrics.
%%   It seems possible that integrating higher powers of the distance function could correspond exactly to higher power metrics.
%%   However, the proof techniques of this paper do not easily extend to that case.
%%   It does seem that as the power goes to infinity, shortest paths converge to the minimum spanning tree.
%%   For other powers in between $2$ and $\infty$, this question remains open.

% \begin{appendices}
  \appendix
    \section{Nearest Neighbor Metric and Edge-Power Metrics relate to Single Linkage Clustering, Level Sets, and k-Centers
  clustering}\label{ap:clustering-link}

  Many popular clustering algorithms, including $k$-centers, $k$-means,
  and $k$-medians clustering, use Euclidean distance as a measure of
  distance between points in $\mathbb{R}^d$. These methods are useful
  when clusters are spherical and well-separated. However, it is believed by
  practitioners that
  data-sensitive distances more accurately capture intrinsic
  distances between data \cite{alamgir12shortest}.

  The celebrated single-linkage clustering algorithm~\cite{Gower1969, Yaro2017},
  which is clustering based on an MST, is a widely used tool in
  machine learning, and gets around many of the problems of the
 Euclidean distance clustering. In single-linkage clustering,
   two points are considered similar if the maximum length
   edge on the path between them in the MST is small. This turns out to be equivalent to 
  computing the normalized $\infty$-power metric between the two
  points. Therefore, single linkage clustering can be seen
  as clustering using the normalized $\infty$-power metric.
  Generally, normalized $p$-power metrics can be seen as an
  intermediary between Euclidean distances ($1$-power metrics) and
  Euclidean MST-based clustering.

  Clustering with $p$-power metric relates to another popular
 clustering method in machine learning, known as level-set clustering.
  Loosely speaking, level set clustering involves finding an estimate
  for the probability density that points are drawn from, finding a
  cut threshold $t$, and then taking as clusters all regions with
  probability density $ > t$. Level set clustering
  has appeared in many incarnations~\cite{Wishart69, Stuetzle2003, Stuetzle2007}, including the celebrated and
  widely used DBScan method~\cite{Ester1996} and its
  considerable number of variations~\cite{OPTICS96}.
It is known
  that level-set clustering is related to single-linkage clustering, as
  the latter is an approximation of the former~\cite{Wishart69,
  Stuetzle2007}.  
  Level-set methods
  have the advantage that they can find arbitrarily shaped
  clusters~\cite{Ester1996}, but can cause two points that are very
  close in Euclidean distance to be considered far apart. 

  Clustering with the $p$-power metric incorporates the
  advantages of both Euclidean distance clustering and level set
  clustering, as it is both data-sensitive and takes into account
  overall
  Euclidean distance between two points.
  Here, $p$ can be toggled to change the sensitivity of the metric
  to the underlying density.
  As the number of samples
  drawn from our probability density grows large, it has been proven that the behavior of
  normalized $p$-power metrics converges to a natural geodesic distance on the
  underlying probability density~\cite{hwang2016}. Clustering with
  this geodesic distance for $p=1$ is exactly Euclidean clustering, and for
  $p=\infty$ is exactly the level set method. 
   Thus, clustering with $p$-power metric converges to
  a clustering method that smoothly interpolates between
  Euclidean-distance clustering and level set clustering. 
% \tim{Maybe you want to mention the self-sparsifying graph when you raise
% $p$ from $1$ to $\infty$}

\section{Proving Faster and Sparser-than-Euclidean Approximate
  Spanners}\label{ap:general-spanner}
In this appendix, we finish the proof of
Theorem~\ref{thm:general-spanner} based on the outline given in
Section~\ref{sec:general-spanner}.

\subsection{$1+O(\delta^2)$ spanners can be generated from a
$1/\delta$ WSPD}

\begin{definition}\label{def:critical} Let $e$ be a \textbf{critical} edge in a
  shortest path metric on any graph if the
  (possibly-not-unique) shortest path between the endpoints of $e$ is the edge $e$.
\end{definition}
\begin{lemma} The set of critical edges on any graph forms a $1$-spanner of the
shortest path metric. \end{lemma}
The above lemma is known in the literature.

% \tim{Make an algorithm box for this process.}
To check that any graph $H$ is a $(1+O(\delta^2)$ spanner of any graph $G$, it suffices
to prove that all critical edges in the edge-squared metric have a stretch
no larger than $1+O( \delta^2)$.
Let $G$ be the edge-squared graph arising from points $P \subset \mathbb{R}^d$. Build a well-separated pair decomposition on P, with pairs given as $\{A_1, B_1\}, \{A_2, B_2\}, \ldots \{A_m, B_m\}$.
Create a spanner $H$ as follows: for each pair $\{A_i, B_i\}$, connect
an edge $\{a, b\}, a\in A_i, b\in B_i$ such that the Euclidean distance
between $a$ and $b$ is a $(1+c\delta^2)$ approximation of the shortest distance
between point sets $A_i$ and $B_i$, for some constant $c$ independent of
$i$. This can be accomplished in $O(1)$ time assuming a preprocessing
step of $O(\delta^{-d} \log\left(\frac{1}{\delta}\right)$ time, as noted
in Callahan's paper on constructing a Euclidean MST~\cite{Callahan1995}. Do this for all $ 1 \leq i \leq m$.

For each critical edge $(s,t)$, consider the well-separated pair $\{A, B\}$ that $(s,t)$ is part of. Let $ s\in A$ and $t \in B$. Let $(a,b)$ be a $(1+c\delta^2)$-approximate shortest edge between $A$ and $B$ ($a \in A, b\in B$).
Scale $||a-b||_2$ to be 1. $A$ and $B$ have Euclidean radius at most
$\delta$, by the definition of a well separated pair. By induction on
Euclidean distance, $H$ is an edge-squared $2$-spanner of the
edge-squared metric for all points in $A$ and $B$ and all points in $B$
(assuming sufficiently small $\delta$).

\begin{lemma}
\begin{align*}
  dist_H(s,t)
    &\leq dist_H(s,a)+dist_H(a,b)+dist_H(b,t)\\
    & \leq 1+O(\delta^2)
\end{align*}
\end{lemma}
\begin{proof}
We know $dist_H(a,b)$ = 1 by our scaling, and $$dist_H(s,a) \leq 2 \cdot (dist_G(s,a)) \leq 2 \cdot ||s-a||^2 \leq 8\delta^2$$
The first inequality follows by the inductive hypothesis that $H$ is a 2-spanner of $G$ in $A$. The third inequality follows since both $s$ and $a$ are contained in a ball of radius $\delta$.

The same bound applies for $dist_H(b,t)$.
\end{proof}

\begin{lemma} \label{lm}
$$(1+c\delta^2)(dist_G(s,t)) \geq dist_G(a,b) = 1$$
$$ \Rightarrow dist_G(s,t)  \geq \frac{1}{1+c\delta^2}$$
\end{lemma}

Lemma \ref{lm} follows from the fact that $(a,b)$ is a $(1+c\delta^2)$ approximate shortest distance between $A$ and $B$.

Therefore

\begin{align*}
  stretch_H(s,t)
    &\leq \frac{dist_H(s,t)}{dist_G(s,t)} \\
    &\leq (1+16\delta^2)(1+c\delta^2) \\
    &= 1+O(\delta^2)
\end{align*}

Thus we have proven that $H$ is a $1+16\delta^2$ spanner. Now set
$\epsilon = \delta^2$, which completes
proof of Theorem~\ref{thm:general-spanner}.

\section{Spanners in the Probability Density Setting: Full
  Proof}\label{ap:distribution-spanner}

We prove Theorem~\ref{thm:distribution-spanner} in full.
Through this section, we assume that $D$ is a probability density
function with support on smooth connected compact manifold with intrinsic dimension
$d$ embedded in ambient space $\mathbb{R}^s$,
with smooth boundary of bounded curvature. This
probability density function is further assumed to be bounded
above and below, and to be Lipschitz. For simplicity, we assume that
$s=d$, and we can prove all our results when $s > d$ by taking
coordinate charts from the manifold into Euclidean space.  We
will show at the end of the section that if the distribution is
supported on a convex set of full dimension in the ambient space,
then the $k$-NN graph is Gabriel for the same $k$. It is not
difficult to see that Gabriel graphs are $1$-spanners of the
edge-squared metric~\cite{SridharMaster}.

\begin{lemma}
  Let $\M$ be a compact object in $\mathbb{R}^d$, whose
  boundary is a smooth manifold of dimension $d-1$ with bounded
  curvature.
  Let $\B$ be any ball with sufficiently small radius
  $r_B$ with center in $\M$, that intersects the boundary of $D$
  at some point $x$.
  Let $H$ be the halfspace tangent to $\M$ at $x$ containing the center
  of the ball.

For any point $Q \in \M$, let $Q'$ be the point in $B$
closest to $Q$. If $d(Q', H) / r_B > c$ for arbitrary constant $c$,
then $d(Q, H) \geq c'$ for some constant $c'$.
\end{lemma}

This is a basic fact about the smoothness and bounded curvature
of the boundary.

\begin{lemma} Pick $n$ points from $D$. W.h.p, any two points in
  $Support(D)$ with Euclidean distance $ \geq \Omega(1)$
  have nearest neighbor metric of $o(1)$.
\end{lemma}

This is implicit in~\cite{hwang2016}.
\begin{lemma} For any ball $\B$ with center $O$ and any point $Q'$ on the
  boundary of $B$, let $B_{Q'O}$ be the ball with diameter $Q'O$.
  Let $H$ be any halfspace containing $O$.
  If $d(Q', H) / r_B \leq c$ for some constant
  $c$ possibly depending on the dimension $d$, then
  $\vol(\B_{Q'O} \cap H) \geq \frac{1-c'}{2^d} \vol(\B \cap H)$
  for some
  constant $c'$, where $c'$ goes to $0$ as $c$ goes to $0$.
\end{lemma}

\begin{proof} First, let us consider the case where $d(Q' H) =
0$, that is, $Q'$ is contained in halfspace $H'$.  In this case,
  dilating $B_{Q'O} \cap H$ by a factor of $2$ about point $Q'$
  gives a superset of $B\cap H$, as $B_{Q'O}$ maps to $B$ and $H$
  maps to a halfspace strictly containing $H$. In this case,
  $\vol(\B_{Q'O} \cap H) \geq \frac{1}{2^d} \vol(\B \cap H)$ as
  desired.
  The case when
  $d(Q', H)/r_B$ is bounded follows in a straightforward manner.
  \end{proof}

This leads us to our following theorem:
\begin{theorem} For any $n$ point set $P$ picked i.i.d from $D$, consider any
  point $O$. Let $\B$ be the $k$-NN ball of $O$.
  Let $Q \in Support(D)$ be any point outside
  $\B$, and let the closest point to $Q$ in $\B$ be $Q'$. For a
  point $x$ inside $B$ on the boundary of $D$ (assuming such a
  point exists), let $H$ be the tangent halfplane containing the
  center of $\B$.

  Then: either $d(Q', H) /
  r_B \leq c'$ for some constant $c'$ or there exists a constant
  $c$ where $|QO| > c$. Here, $c$ and $c'$ are
  independent of the number of points chosen, and $c'$ can be set
  arbitrarily small.

  In the latter case, w.h.p. $QO$ is not in the edge-squared
  $1$-spanner. In the former case, setting $c'$ to be a very
  small constant $\epsilon$ lets us say:

  \begin{align} \label{eq:vol}
    \vol(\B_{Q'O} \cap H) \geq \frac{1-\epsilon}{2^d} \vol(\B
    \cap H),
  \end{align}
  or equivalently:
  \begin{align} \label{eq:prob}
    & \prob{x\sim D}{x \in \B_{QO} | x \in \B}
  \\
    \label{eq:prob-2}
    \geq & \prob{x\sim D}{x \in \B_{Q'O} | x \in \B}
  \\ \label{eq:prob-3}
    \geq & \frac{1-\eps-o(1)}{2^d}
  \end{align}
\end{theorem}
Expression~\ref{eq:prob-2} $>$ Expression~\ref{eq:prob-3} follows from Equation~\ref{eq:vol}, and
the fact that the radius of the $k$-NN ball goes to $0$ as $n$
gets large, and thus the probability density of sampling $x$ from
$D$ conditioned on $x$ being in $\B$ approaches the uniform
density in $\B \cap Support(D)$. Also, $B \cap H$ approaches $B \cap Support(D)$ as the
radius of $B$ goes to $0$.

Expression~\ref{eq:prob} $>$
Expression~\ref{eq:prob-2} since $\B_{QO} \supset B_{Q'O}$.
(Here, the $k$-NN ball $B$ w.r.t. point $O$ is defined as the ball
centered at $O$ with radius equal to the distance of the $k^{th}$
nearest neighbor to $O$).

Note that the $k-1$ nearest neighbors of $O$, conditioned only on
the radius of $B$, are distributed
equivalently to $k-1$ i.i.d samples of $D$ conditioned on containment
in $\B$. It follows
that for any point $Q$ outside $B$ and in the support of $D$,
where $|QO|< c$:

\begin{align*}
  \prob{P \sim D^k}{QO \text{ is not Gabriel w.r.t. $P$} | Q \not\in B}\\
  \geq 1- \left(1 - \frac{1-\eps-o(1) }{2^d} \right)^k
\end{align*}

Thus, setting $\epsilon = 0.1$ and $ k > O(\log n/2^d) $, and
factoring in the case where $|QO| > c$, then w.h.p.:
\[
  \prob{P \sim D^k}{QO \text{ is not critical w.r.t. $P$} | Q \not\in B}
  \]
Here, we recall that an edge $AB$ is Gabriel with respect to a
point set $P$ if and only if $\B_{AB}$ does not contain any points
in $P$. Note that every non-Gabriel edge is non-critical, where a
critical edge is an edge that must be in the $1$-spanner (as
in Definition~\ref{def:critical}).
 Thus taking the union bound over $Q, O \in P$ gives us that
no edge outside the $k$-NN graph is critical w.h.p, and thus the
$k$-NN graph contains all critical edges and is a $1$-spanner
w.h.p.

This proves Theorem~\ref{thm:distribution-spanner} when the
support of $D$ has the same intrinsic dimension as the ambient
space. If the support of $D$ has  dimension $d < d'$ (where $d'$
is the ambient dimension of the space), simply take coordinate charts from
$D$ onto $\mathbb{R}^d$ and the previous arguments will still
carry through
.
We should note that if no point $x$ inside $B$ on the boundary of
$D$ exists, then we can ignore $H$ and all the steps of
the proof still follow.

% \end{appendices}

   % \input{definitions} %% now in introduction
  %% \input{NN}
  %% \input{general-spanner}
  %% \input{distribution-spanner}
  %% \input{edge-power}
  %% \input{conclusions}
  %% \appendix
  %% %\input{combined-edge-squared}
  %% \input{motivation}
  %% %\input{higher_powers}

  \clearpage
   \bibliographystyle{alpha}
%   \nocite{alamgir12shortest,hein07graph}
   %   \bibliography{Bibtex/bibtexNN,Bibtex/bibtex}
      \bibliography{Bibtex/NNbibtex,Bibtex/bibtex}
   % \bibliography{Bibtex/bibtexNN}

\end{document}